\documentclass[english]{amsart}
\usepackage{mathptmx}
\usepackage[T1]{fontenc}
\usepackage{amsthm}
\usepackage{graphicx}
\usepackage{amssymb}

\newcommand{\noun}[1]{\textsc{#1}}
\providecommand{\tabularnewline}{\\}

\numberwithin{equation}{section}
\numberwithin{figure}{section}
\newtheorem{thm}{Theorem}
  \newtheorem{example}[thm]{Example}
  \newtheorem{notation}[thm]{Notation}
  \newtheorem{defn}[thm]{Definition}
  \newtheorem{prop}[thm]{Proposition}
  \newtheorem{algorithm}[thm]{Algorithm}
  \newtheorem{condition}[thm]{Condition}
  \newtheorem{rem}[thm]{Remark}

\begin{document}

\title{CAD Adjacency Computation Using Validated Numerics}

\author{Adam Strzebo\'nski}

\address{Wolfram Research Inc., 100 Trade Centre Drive, Champaign, IL 61820,
U.S.A. }

\email{adams@wolfram.com}

\date{April 21, 2017}
\begin{abstract}
We present an algorithm for computation of cell adjacencies for well-based
cylindrical algebraic decomposition. Cell adjacency information can
be used to compute topological operations e.g. closure, boundary,
connected components, and topological properties e.g. homology groups.
Other applications include visualization and path planning. Our algorithm
determines cell adjacency information using validated numerical methods
similar to those used in CAD construction, thus computing CAD with
adjacency information in time comparable to that of computing CAD
without adjacency information. We report on implementation of the
algorithm and present empirical data.
\end{abstract}
\maketitle

\section{Introduction}

A semialgebraic set is a subset of $\mathbb{R}^{n}$ which is a solution
set of a system of polynomial equations and inequalities. Computation
with semialgebraic sets is one of the core subjects in computer algebra
and real algebraic geometry. A variety of algorithms have been developed
for real system solving, satisfiability checking, quantifier elimination,
optimization and other basic problems concerning semialgebraic sets
\cite{C,BPR,CJ,DSW,GV,HS,LW,R,T,W1}. Every semialgebraic set can
be represented as a finite union of disjoint cells bounded by graphs
of semialgebraic functions. The Cylindrical Algebraic Decomposition
(CAD) algorithm \cite{C,CJ,S7} can be used to compute a cell decomposition
of any semialgebraic set presented by a quantified system of polynomial
equations and inequalities. Alternative methods of computing cell
decompositions are given in \cite{CMXY,S11,S12}. For solving certain
problems, for instance computing topological properties or visualization,
it is not sufficient to know a cell decomposition of the set, but
it is also necessary to know how the cells are connected together.
\begin{example}
The CAD algorithm applied to the equation $y^{2}=x(x^{4}-1)$ gives
four one-dimensional cells and three zero-dimensional cells shown
in Figure \ref{fig}. To find the connected components of the solution
set it is sufficient to know which one-dimensional cells are adjacent
to which zero-dimensional cells.

\begin{figure}
\includegraphics[width=2\columnwidth, trim = 0mm 5mm 0mm 220mm, clip]{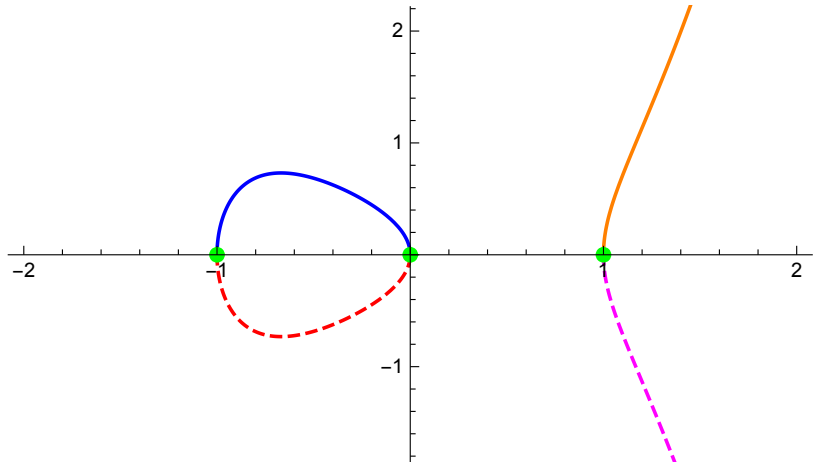}

\caption{\label{fig}$y^{2}=x(x^{4}-1)$}
\end{figure}

\end{example}
Several algorithms for computing cell adjacencies have been developed.
The algorithm given in \cite{SS} computes cell adjacencies for CAD
that are well-based. A CAD is well-based if none of the polynomials
whose roots appear in cell description vanishes identically at any
point. This is a somewhat stronger condition than well-orientedness
required for the McCallum projection \cite{MC2}, nevertheless a large
portion of examples that appear in practice satisfies the condition.
In a well-based CAD all cell adjacencies can be determined from adjacencies
between cells whose dimensions differ by one. In $\mathbb{R}^{2}$
all CADs are well-based. Algorithm for computing cell adjacencies
in arbitrary CADs in $\mathbb{R}^{3}$ has been given in \cite{ACM2}.
For determining cell adjacencies \cite{SS} proposes methods based
on fractional power series representations of polynomial roots. Another
method, given in \cite{MCC}, computes adjacencies between a zero-dimensional
cell and one-dimensional cells by analyzing intersections of the one
dimensional cells with sides of a suitable box around the zero-dimensional
cell. For an alternative method of computing connectivity properties
of semialgebraic sets see \cite{CA,BPR1,BPR,SSC}.

In this paper we present a new algorithm which computes cell adjacencies
for well-based CAD. The algorithm uses validated numerical methods
similar to those used in \cite{S7} for construction of CAD cell sample
points. The method is based on computation of approximations of polynomial
roots and increasing the precision of computations until validation
criteria are satisfied. Unlike the previously known algorithms, it
does not require polynomial computations over algebraic number fields
or computation with fractional power series representations of polynomial
roots. Also, unlike the CAD construction algorithm given in \cite{S7},
the algorithm never needs to revert to exact algebraic number computations.
We have implemented the algorithm as an extension to the CAD implementation
in \emph{Mathematica}. Empirical results show that computation of
CAD with cell adjacency data takes time comparable to computation
of CAD without cell adjacency data.

The general idea of the algorithm is as follows. It starts, similarly
as the CAD algorithm, with computing a sample point in each cell in
$\mathbb{R}^{k}$ for all $k\leq n$. The sample point of a cell in
$\mathbb{R}^{k+1}$ extends the sample point of the projection of
the cell on $\mathbb{R}^{k}$. Then for each pair of adjacent CAD
cells $C$ and $C^{\prime}$ in $\mathbb{R}^{k}$ with $\dim C^{\prime}=\dim C-1$
the algorithm constructs a point $p\in C$ that is {}``sufficiently
close'' to the sample point $p^{\prime}$ of $C^{\prime}$. Here
{}``sufficiently close'' means that computing approximations of
roots of projection polynomials at $p$ and $p^{\prime}$ is sufficient
to identify which roots over $C$ tend to which roots over $C^{\prime}$
and to continue the construction to pairs of adjacent CAD cells in
$\mathbb{R}^{k+1}$. The construction gives all pairs of adjacent
cells in $\mathbb{R}^{n}$ whose dimensions differ by one. For well-based
CAD this is sufficient to determine all cell adjacencies.

\section{Preliminaries}

A \emph{system of polynomial equations and inequalities} in variables
$x_{1},\ldots,x_{n}$ is a formula\[
S(x_{1},\ldots,x_{n})=\bigvee_{1\leq i\leq l}\bigwedge_{1\leq j\leq m_{i}}f_{i,j}(x_{1},\ldots,x_{n})\rho_{i,j}0\]
where $f_{i,j}\in\mathbb{R}[x_{1},\ldots,x_{n}]$, and each $\rho_{i,j}$
is one of $<,\leq,\geq,>,=,$ or $\neq$. 

A subset of $\mathbb{R}^{n}$ is \emph{semialgebraic} if it is a solution
set of a system of polynomial equations and inequalities. 

A \emph{quantified system of real polynomial equations and inequalities}
in free variables $x_{1},\ldots,x_{m}$ and quantified variables $x_{m+1},\ldots,x_{n}$
is a formula\begin{equation}
Q_{1}x_{m+1}\ldots Q_{n-m}x_{n}S(x_{1},\ldots,x_{n})\label{quantsyst}\end{equation}
 Where $Q_{i}$ is $\exists$ or $\forall$, and $S$ is a system
of real polynomial equations and inequalities in $x_{1},\ldots,x_{n}$.

By Tarski's theorem (see \cite{T}), solution sets of quantified systems
of real polynomial equations and inequalities are semialgebraic.
\begin{notation}
For $k\geq1$, let $\overline{a}$ denote a $k$-tuple $(a_{1},\ldots,a_{k})$
of real numbers and let $\overline{x}$ denote a $k$-tuple $(x_{1},\ldots,x_{k})$
of variables.
\end{notation}
Every semialgebraic set can be represented as a finite union of disjoint
\emph{cells} (see \cite{L}), defined recursively as follows.
\begin{enumerate}
\item A cell in $\mathbb{R}$ is a point or an open interval.
\item A cell in $\mathbb{R}^{k+1}$ has one of the two forms\begin{eqnarray*}
 & \{(\overline{a},a_{k+1}):\overline{a}\in C_{k}\wedge a_{k+1}=r(\overline{a})\}\\
 & \{(\overline{a},a_{k+1}):\overline{a}\in C_{k}\wedge r_{1}(\overline{a})<a_{k+1}<r_{2}(\overline{a})\}\end{eqnarray*}
where $C_{k}$ is a cell in $\mathbb{R}^{k}$, $r$ is a continuous
semialgebraic function, and $r_{1}$ and $r_{2}$ are continuous semialgebraic
functions, $-\infty$, or $\infty$, and $r_{1}<r_{2}$ on $C_{k}$. 
\end{enumerate}
For a cell $C\subseteq\mathbb{R}^{n}$ let $\Pi_{k}(C)\subseteq\mathbb{R}^{k}$,
for $k\leq n$, denote the projection of $C$ on $\mathbb{R}^{k}$.
A finite collection $D$ of cells in $\mathbb{R}^{n}$ is \emph{cylindrically
arranged} if for any $C_{1},C_{2}\in D$ and any $k\leq n$ $\Pi_{k}(C_{1})$
and $\Pi_{k}(C_{2})$ are either disjoint or identical. 

A \emph{cylindrical algebraic decomposition (CAD) of $\mathbb{R}^{n}$
}is a finite collection $D$ of pairwise disjoint cylindrically arranged
cells in $\mathbb{R}^{n}$ such that $\bigcup_{C\in D}C=\mathbb{R}^{n}$. 

Let $P\subset\mathbb{R}[x_{1},\ldots,x_{n}]$ be a finite set of polynomials.
A CAD $D$ of\emph{ $\mathbb{R}^{n}$ }is \emph{P-invariant }if each
element of $P$ has a constant sign on each cell of $D$.

Let \emph{$A\subseteq\mathbb{R}^{n}$ }be a semialgebraic set. A CAD
$D$ of\emph{ $\mathbb{R}^{n}$ }is \emph{consistent with} $A$ if
$A=\bigcup_{C\in D_{A}}C$ for some $D_{A}\subseteq D$.

Let $C_{1},C_{2}\in D$. $C_{1}$ and $C_{2}$ are \emph{adjacent}
if $C_{1}\neq C_{2}$ and $C_{1}\cup C_{2}$ is connected.

For a semialgebraic set $A$ presented by a quantified system of polynomial
equations and inequalities (\ref{quantsyst}), the CAD algorithm can
be used to find a CAD $D$ of $\mathbb{R}^{n}$ consistent with $A$.
The CAD $D$ is represented by a cylindrical algebraic formula (CAF).
A CAF describes each cell by giving explicit semialgebraic function
bounds and the Boolean structure of a CAF reflects the cylindrical
arrangement of cells. Before we give a formal definition of a CAF,
let us first introduce some terminology.

Let $k\geq1$ and let $f=c_{d}y^{d}+\ldots+c_{0}$, where $c_{0},\ldots,c_{d}\in\mathbb{\mathbb{Z}}[\overline{x}]$.
A \emph{semialgebraic function} given by the \emph{defining polynomial}
$f$ and a \emph{root number} $\lambda\in\mathbb{N}_{+}$ is the function\begin{equation}
Root_{y,\lambda}f:\mathbb{R}^{k}\ni\overline{a}\longrightarrow Root_{y,\lambda}f(\overline{a})\in\mathbb{R}\label{rootfun}\end{equation}
where $Root_{y,\lambda}f(\overline{a})$ is the $\lambda$-th real
root of $f(\overline{a},y)\in\mathbb{R}[y]$. The function is defined
for those values of $\overline{a}$ for which $f(\overline{a},y)$
has at least $\lambda$ real roots. The real roots are ordered by
the increasing value and counted with multiplicities. A real algebraic
number $Root_{y,\lambda}f\in\mathbb{R}$ given by a \emph{defining
polynomial} $f\in\mathbb{Z}[y]$ and a \emph{root number} $\lambda$
is the $\lambda$-th real root of $f$. 

Let $C$ be a connected subset of $\mathbb{R}^{k}$. The function
$Root_{y,\lambda}f$ is\emph{ regular} on \emph{$C$} if it is continuous
on $C$, $c_{d}(\overline{a})\neq0$ for all $\overline{a}\in C$,
and there exists\emph{ }$m\in\mathbb{\mathbb{N}}_{+}$ such that for
any $\overline{a}\in C$ $Root_{y,\lambda}f(\overline{a})$ is a root
of $f(\overline{a},y)$ of multiplicity $m$. 

The polynomial $f$ is \emph{degree-invariant} on $C$ if there exists\emph{
}$e\in\mathbb{\mathbb{N}}$ such that $c_{d}(\overline{a})=\ldots=c_{e+1}(\overline{a})=0\wedge c_{e}(\overline{a})\neq0$
for all $\overline{a}\in C$. 

A set $W=\{f_{1},\ldots,f_{m}\}$ of polynomials is \emph{delineable}
on $C$ if all elements of $W$ are degree-invariant on $C$ and for
$1\leq i\leq m$\[
f_{i}^{-1}(0)\cap(C\times\mathbb{R})=\{r_{i,1},\ldots,r_{i,l_{i}}\}\]
where $r_{i,1},\ldots,r_{i,l_{i}}$ are disjoint regular semialgebraic
functions and for $i_{1}\neq i_{2}$ $r_{i_{1},j_{1}}$ and $r_{i_{2},j_{2}}$
are either disjoint or equal. Functions $r_{i,j}$ are \emph{root
functions of $f_{i}$ over $C$}.

Let $W$ be delineable on $C$, let $r_{1}<\ldots<r_{l}$ be all root
functions of elements of $W$ over $C$, and let $r_{0}=-\infty$
and $r_{l+1}=\infty$. For $1\leq i\leq l$, the $i$-th \emph{$W$-section
over $C$} is the set\[
\{(\overline{a},a_{k+1}):\overline{a}\in C\wedge a_{k+1}=r_{i}(\overline{a})\}\]
For $1\leq i\leq l+1$, the $i$-th \emph{$W$-sector over $C$} is
the set\[
\{(\overline{a},a_{k+1}):\overline{a}\in C\wedge r_{i-1}(\overline{a})<a_{k+1}<r_{i}(\overline{a})\}\]
\emph{$W$-stack over $C$} is the set of all $W$-sections and $W$-sectors
over $C$.

A formula $F$ is an \emph{algebraic constraint} with \emph{bounds}
$BDS(F)$ if it is a level-$k$ equational or inequality constraint
with $1\leq k\leq n$ defined as follows.\emph{ }
\begin{enumerate}
\item \emph{A level}-$1$ \emph{equational constraint} has the form $x_{1}=r$,
where $r$ is a real algebraic number, and $BDS(F)=\{r\}$.
\item \emph{A level}-$1$ \emph{inequality constraint} has the form $r_{1}<x_{1}<r_{2}$,
where $r_{1}$ and $r_{2}$ are real algebraic numbers, $-\infty$,
or $\infty$, and $BDS(F)=\{r_{1},r_{2}\}\setminus\{-\infty,\infty\}$. 
\item \emph{A level}-$k+1$ \emph{equational constraint} has the form $x_{k+1}=r(\overline{x})$,
where $r$ is a semialgebraic function, and $BDS(F)=\{r\}$.
\item \emph{A level}-$k+1$ \emph{inequality constraint} has the form $r_{1}(\overline{x})<x_{k+1}<r_{2}(\overline{x})$,
where $r_{1}$ and $r_{2}$ are semialgebraic functions, $-\infty$,
or $\infty$, and $BDS(F)=\{r_{1},r_{2}\}\setminus\{-\infty,\infty\}$. 
\end{enumerate}
A level-$k+1$ algebraic constraint $F$ is \emph{regular} on a connected
set $C\subseteq\mathbb{R}^{k}$ if all elements of $BDS(F)$ are regular
on $C$ and, if $F$ is an inequality constraint, $r_{1}<r_{2}$ on
$C$.
\begin{defn}
An \emph{atomic cylindrical algebraic formula (CAF)} $F$ in $(x_{1},\ldots,x_{n})$
has the form $F_{1}\wedge\ldots\wedge F_{n}$, where $F_{k}$ is a
level-$k$ algebraic constraint for $1\leq k\leq n$ and $F_{k+1}$
is regular on the solution set of $F_{1}\wedge\ldots\wedge F_{k}$
for $1\leq k<n$. 

\emph{Level-$k$ cylindrical formulas} in $(x_{1},\ldots,x_{n})$
are defined recursively as follows
\begin{enumerate}
\item A level-$n$ cylindrical formula is $false$ or a disjunction of level-$n$
algebraic constraints.
\item A level-$k$ cylindrical formula, with $1\leq k<n$, is $false$ or
has the form\[
(F_{1}\wedge G_{1})\vee\ldots\vee(F_{m}\wedge G_{m})\]
where $F_{i}$ are level-$k$ algebraic constraints and $G_{i}$ are
level-$k+1$ cylindrical formulas.
\end{enumerate}
A \emph{cylindrical algebraic formula (CAF)} is a level-$1$ cylindrical
formula $F$ such that distributing conjunction over disjunction in
$F$ gives \[
DNF(F)=F_{1}\vee\ldots\vee F_{l}\]
where each $F_{i}$ is an atomic CAF. Let $C(F_{i})$ denote the solution
set of $F_{i}$ and let $D(F)=\{C(F_{1}),\ldots,C(F_{l})\}$. The
\emph{bound polynomials }of $F$ is a finite set $BP(F)\subset\mathbb{R}[x_{1},\ldots,x_{n}]$
which consists of all polynomials $f$ such that $Root_{x_{k},\lambda}f\in BDS(G)$
for some $1\leq k\leq n$ and a level-$k$ algebraic constraint $G$
that appears in $F$. 
\end{defn}
Note that $C(F_{i})$ is a cell and $D(F)$ is a finite collection
of pairwise disjoint cylindrically arranged cells. 

For a CAF $F$ in $(x_{1},\ldots,x_{n})$, let $\Pi_{k}(F)$ denote
the CAF in $(x_{1},\ldots,x_{k})$ obtained from $F$ by removing
all level-$k+1$ subformulas. Then \[
D(\Pi_{k}(F))=\{\Pi_{k}(C)\::\: C\in D(F)\}\]

Following the terminology of \cite{SS}, we define a well-based CAF
as follows.
\begin{defn}
A CAF\emph{ }$F$ is \emph{well-based} if\emph{ $D(F)$} is a $BP(F)$-invariant
CAD of $\mathbb{R}^{n}$ and for any $f\in BP(F)$ if $f\in\mathbb{R}[x_{1},\ldots,x_{k+1}]\setminus\mathbb{R}[x_{1},\ldots,x_{k}]$
then for any $\overline{a}\in\mathbb{R}^{k}$ $f(\overline{a},x_{k+1})$
is not identically zero. 
\end{defn}
In a CAD corresponding to a well-based CAF a closure of a cell is
a union of cells and the only cells from other stacks adjacent to
a given section are sections defined by the same polynomial. Moreover,
any two adjacent cells have different dimensions and are connected
through a chain of adjacent cells with dimensions increasing by one,
and hence to determine all cell adjacencies it is sufficient to find
all pairs of adjacent cells whose dimensions differ by one. These
properties, stated precisely in Proposition \ref{WellBasedProp},
are essential for our algorithm. 
\begin{prop}
\label{WellBasedProp}Let $F$ be a well-based CAF.
\begin{enumerate}
\item If $C\in D(F)$ then there exits cells $C_{1},\ldots,C_{m}\in D(F)$
such that $\overline{C}=C\cup C_{1}\cup\ldots\cup C_{m}$.
\item Let $C\in D(F)$ be a section\[
C=\{(\overline{a},a_{n}):\overline{a}\in\Pi_{n-1}(C)\wedge a_{n}=Root_{x_{n},\lambda}f(\overline{a})\}\]
and let $C^{\prime}\in D(\Pi_{n-1}(F))$ be a cell adjacent to $\Pi_{n-1}(C)$
with $\dim C^{\prime}<\dim\Pi_{n-1}(C)$. Then either $\overline{C}\cap(C^{\prime}\times\mathbb{R})$
is equal to a section\[
\{(\overline{a},a_{n}):\overline{a}\in C^{\prime}\wedge a_{n}=Root_{x_{n},\lambda^{\prime}}f(\overline{a})\}\]
for some $1\leq\lambda^{\prime}\leq\deg_{x_{n}}(f)$, or for any $\overline{a}\in C^{\prime}$
\[
\lim_{\overline{b}\in\Pi_{n-1}(C),\overline{b}\rightarrow\overline{a}}Root_{x_{n},\lambda}f(\overline{a})=-\infty\]
 or for any $\overline{a}\in C^{\prime}$ \[
\lim_{\overline{b}\in\Pi_{n-1}(C),\overline{b}\rightarrow\overline{a}}Root_{x_{n},\lambda}f(\overline{a})=\infty\]

\item Let $C_{k},C_{l}\in D(F)$ be adjacent cells such that $\dim(C_{k})=k$
and $\dim(C_{l})=l$. Then $k\neq l$ and if $k<l$ then there exist
cells $C_{k+1},\ldots,C_{l-1}\in D(F)$ such that $\dim(C_{j})=j$
and $C_{j}\subseteq\overline{C_{j+1}}$ for $k\leq j<l$.
\end{enumerate}
\end{prop}
\begin{proof}
Part $(1)$ is Lemma 1 of \cite{SS}. To prove $(2)$ first note that,
by $(1)$, $C^{\prime}\subset\overline{\Pi_{n-1}(C)}$. By Lemma 2
of \cite{SS}, there exists a unique continuous function $r:\overline{\Pi_{n-1}(C)}\rightarrow\mathbb{R}\cup\{-\infty,\infty\}$
extending $Root_{x_{n},\lambda}f(\overline{a})$. Moreover, $r$ is
either infinite or a root of $f$. Since $f$ is delineable on $C^{\prime}$
and $C^{\prime}$ is connected, either $r$ is a root of $f$ on $C^{\prime}$,
or $r\equiv-\infty$ on $C^{\prime}$, or $r\equiv\infty$ on $C^{\prime}$. 

We will prove $(3)$ by induction on $n$. Note that by $(1)$ it
is sufficient to find cells $C_{j}$ such that $\dim(C_{j})=j$ and
$C_{j}$ and $C_{j+1}$ are adjacent for $k\leq j<l$. If $n=1$ then
dimensions of any pair of adjacent cells differ by one, hence $(3)$
is true. To prove $(3)$ for $n>1$ we will use induction on $l-k$.
If $l-k=1$ then $(3)$ is true. Suppose $l-k>1$. Let $C_{k^{\prime}}^{\prime}=\Pi_{n-1}(C_{k})$
and $C_{l^{\prime}}^{\prime}=\Pi_{n-1}(C_{l})$, where $\dim(C_{k^{\prime}}^{\prime})=k^{\prime}$
and $\dim(C_{l^{\prime}}^{\prime})=l^{\prime}$. If $C_{l}$ is a
section, then, by Lemma 2 of \cite{SS}, there exists a continuous
function $r:\overline{C_{l^{\prime}}^{\prime}}\rightarrow\mathbb{R}\cup\{-\infty,\infty\}$
such that $C_{l}=\{(x,r(x))\::\: x\in C_{l^{\prime}}^{\prime}\}$
and $r$ is infinite or a root of an element of $BP(F)$. In this
case set $s=r$. Similarly, if $C_{l}$ is a sector, then, by Lemma
2 of \cite{SS}, there exists continuous functions $r,s:\overline{C_{l^{\prime}}^{\prime}}\rightarrow\mathbb{R}\cup\{-\infty,\infty\}$
such that $C_{l}=\{(x,y)\::\: x\in C_{l^{\prime}}^{\prime}\wedge r(x)<y<s(x)\}$
and $r$ and $s$ are infinite or roots of elements of $BP(F)$. Since
$l-k>1$, $k^{\prime}<l^{\prime}$. Suppose first that $l^{\prime}-k^{\prime}=1$.
Since $l-k>1$, $C_{k}$ is a section and $C_{l}$ is a sector. If
$C_{k}=\{(x,t(x))\::\: x\in C_{k^{\prime}}^{\prime}\}$, where $t=r$
or $t=s$, then $t$ is finite on $C_{l^{\prime}}^{\prime}$, $C_{k}$
is adjacent to $C_{k+1}:=\{(x,t(x))\::\: x\in C_{l^{\prime}}^{\prime}\}$,
and $C_{k+1}$ is adjacent to $C_{l}$. Since $l=k+2$, $(3)$ is
true. Otherwise, let $C_{k+1}$ be the sector directly below $C_{k}$.
Then $C_{k+1}\subset\{(x,y)\::\: x\in C_{k^{\prime}}^{\prime}\wedge r(x)<y<s(x)\}$,
and hence $C_{k+1}$ is adjacent to $C_{l}$. Again, since $l=k+2$,
$(3)$ is true. Now suppose that $l^{\prime}-k^{\prime}>1$. $\Pi_{n-1}(F)$
is well-based, hence, by the inductive hypothesis on $n$, there exist
cells $C_{k^{\prime}+1}^{\prime},\ldots,C_{l^{\prime}-1}^{\prime}\in D(\Pi_{n-1}(F))$
such that $\dim(C_{j}^{\prime})=j$ and $C_{j}^{\prime}\subseteq\overline{C_{j+1}^{\prime}}$
for $k^{\prime}\leq j<l^{\prime}$. Let $x_{0}\in C_{k^{\prime}}^{\prime}$
and $(x_{0},y_{0})\in C_{k}$. Then $r(x_{0})\leq y_{0}\leq s(x_{0})$.
Since $C_{k^{\prime}}^{\prime}$ is adjacent to $C_{l^{\prime}-1}^{\prime}$,
there exist a sequence $\{x_{n}\}_{n\geq1}$ such that $x_{n}\in C_{l^{\prime}-1}^{\prime}$
and $\lim_{n\rightarrow\infty}x_{n}=x_{0}$. Put $y_{n}=\max(r(x_{n}),\min(s(x_{n}),y_{0}))$.
Then $\lim_{n\rightarrow\infty}(x_{n},y_{n})=(x_{0},y_{0})$. The
set \[
S=\{(x,y)\::\: x\in C_{l^{\prime}-1}^{\prime}\wedge y\in\mathbb{R}\wedge r(x)\leq y\leq s(x)\}\]
is a union of a finite number of cells, $S\subset\overline{C_{l}}$,
and $(x_{n},y_{n})\in S$. Hence, there exists a cell $C\subseteq S$
such that $C$ contains infinitely many elements of the sequence $\{(x_{n},y_{n})\}_{n\geq1}$.
Therefore, $C$ is adjacent to both $C_{k}$ and $C_{l}$. Since $\dim C-k<l-k$
and $l-\dim C<l-k$, $(3)$ is true by the inductive hypothesis on
$l-k$.
\end{proof}
For a given semialgebraic set $A$ a well-based CAF $F$ such that
$D(F)$ is consistent with $A$ may not exist in a given system of
coordinates. However, as shown in \cite{SS}, there always exists
a linear change of variables after which a well-based CAF $F$ such
that $D(F)$ is consistent with $A$ does exist.
\begin{example}
If $A$ is the real solution set of $xy+xz+yz=0$ then a well-based
CAF $F$ such that $D(F)$ is consistent with $A$ does not exist
for any order of variables. A CAD computed using McCallum's projection
operator \cite{MC1} includes cells \begin{eqnarray*}
C_{1} & = & \{(x,y,z)\::\: x>0\wedge y>-x\wedge z=-\frac{xy}{x+y}\}\\
C_{2} & = & \{(x,y,z)\::\: x=0\wedge y=0\}\end{eqnarray*}
$\overline{C_{1}}$ is not a union of cells, since $\overline{C_{1}}\cap C_{2}=\{(x,y,z)\::\: x=0\wedge y=0\wedge z\geq0\}$,
and section $C_{1}$ is adjacent to a sector $C_{2}$ from a different
stack. After the linear change of variables $(x,y,z)\rightarrow(x,y+z,z)$
$A$ is transformed to the solution set of $z^{2}+z(y+2x)+xy=0$.
The following CAF $F$ is well-based and $D(F)$ is consistent with
the transformed $A$.\begin{eqnarray*}
F & = & (x<0\wedge G_{1})\vee(x=0\wedge((y<Root_{y,1}g\wedge G_{1})\vee\\
 &  & (y=Root_{y,1}g\wedge G_{2})\vee(y>Root_{y,1}g\wedge G_{1})))\vee(x>0\wedge G_{1})\end{eqnarray*}
where\begin{eqnarray*}
f & = & z^{2}+z(y+2x)+xy\\
g & = & y^{2}+4x^{2}\\
G_{1} & = & z<Root_{z,1}f\vee z=Root_{z,1}f\vee Root_{z,1}f<z<Root_{z,2}f\vee\\
 &  & z=Root_{z,2}f\vee z>Root_{z,2}f\\
G_{2} & = & z<Root_{z,1}f\vee z=Root_{z,1}f\vee z>Root_{z,1}f\end{eqnarray*}

\end{example}

\section{Root isolation algorithms}

In this section we describe root isolation algorithms we will use
in the algorithm computing cell adjacencies. Let us first introduce
some notations and subalgorithms. 

Let $\Delta(c,r)=\{z\in\mathbb{C}\::\:\mid z-c\mid\leq r\}$ denote
a disk in the complex plane, let $\mathbb{Q}_{2}=\mathbb{Z}[\frac{1}{2}]=\{a2^{b}\::\: a,b\in\mathbb{Z}\}$
denote the set of dyadic rational numbers, and let $I_{2}(\mathbb{C})=\{\Delta(c,r)\::\: c\in\mathbb{Q}_{2}[\imath]\wedge r\in\mathbb{Q}_{2}\wedge r>0\}$
denote the set of discs in the complex plane with dyadic Gaussian
rational centers and dyadic rational radii. For a disc $Z=\Delta(c,r)\in I_{2}(\mathbb{C})$,
let $\gamma(Z):=c$ and $\rho(Z):=r$ denote the center and the radius
of $Z$, let $\underline{Z}:=\max(0,\mid c\mid-r)$ and $\overline{Z}:=\mid c\mid+r$
denote the minimum and maximum of absolute values of elements of $Z$,
let $conj(Z)$ denote the disc that consists of complex conjugates
of elements of $Z$, and let $dbl(Z)=\Delta(c,2r)$ and $quad(Z)=\Delta(c,4r)$.
When we refer to interval arithmetic operations we mean circular complex
interval (disc) arithmetic (see e.g. \cite{PP}). 
\begin{prop}
\label{ApproxRootsProp}There exists an algorithm (ApproximateRoots)
that takes as input a polynomial\[
g=b_{N}x^{N}+\ldots+b_{0}\in\mathbb{\mathbb{Q}}_{2}[\imath][x]\]
and $p\in\mathbb{\mathbb{N}}$ and outputs $(s_{1},\ldots,s_{N})\in\mathbb{\mathbb{Q}}_{2}[\imath]^{N}$
such that for any polynomial \[
f=a_{N}x^{N}+\ldots+a_{0}=a_{N}(x-\sigma_{1})\cdots(x-\sigma_{N})\in\mathbb{C}[x]\]
and any $\epsilon>0$ there exits $p_{0}\in\mathbb{N}$ such that
if $p\geq p_{0}$ and, for all $0\leq i\leq N$, \[
\mid b_{i}-a_{i}\mid\leq2^{-p}\max_{0\leq i\leq N}|a_{i}|\]
then, after a suitable reordering of roots, for all $1\leq j\leq N$
$\mid s_{j}-\sigma_{j}\mid\leq\epsilon$. \end{prop}
\begin{proof}
The algorithm described in \cite{P1} satisfies Proposition \ref{ApproxRootsProp}.
\end{proof}
Let us now describe a subalgorithm computing roots of polynomials
with complex disc coefficients. The algorithm is based on the following
proposition (\cite{S7}, Proposition 4.1).
\begin{prop}
\label{RootMultProp}Let $f\in\mathbb{C}[z]$ be a polynomial of degree
$N$, $z_{0}\in\mathbb{C}$, $r>0$, and let $c_{i}:=\mid\frac{f^{(i)}(z_{0})}{i!}\mid$.
Suppose that \[
\max{}_{0\leq i<k}(\frac{Nc_{i}}{c_{k}})^{\frac{1}{k-i}}<r<\min{}_{k<i\leq N}(\frac{c_{k}}{Nc_{i}})^{\frac{1}{i-k}}\]
Then $f$ has exactly $k$ roots, multiplicities counted, in the disc
$\Delta(z_{0},r)$.
\end{prop}
The following is an extended version of Algorithm 4.2 from \cite{S7}.
The key difference is that this version does not assume that the leading
coefficient does not contain zero and provides a lower bound on the
absolute value of roots that tend to infinity when the leading coefficients
that contain zero vanish.
\begin{algorithm}
\label{IntervalRoots}(IntervalRoots)\\
Input: $Z_{0},\ldots,Z_{N}\in I_{2}(\mathbb{C})$. \\
Output: $D_{1},\ldots,D_{m}\in I_{2}(\mathbb{C})$, positive integers
$k_{1},\ldots,k_{m}$, and a positive radius $R$, such that for any
$a_{0}\in Z_{0},\ldots,a_{N}\in Z_{N}$ and any $1\leq i\leq m$ the
polynomial $f=a_{N}z^{N}+\ldots+a_{0}$ has exactly $k_{i}$ roots
in the disc $D_{i}$, multiplicities counted, and $f$ has no roots
in $\Delta(0,R)\setminus\bigcup_{i=1}^{m}D_{i}$. Moreover, for any
$1\leq i<j\leq m$, $D_{i}\cap D_{j}=\emptyset$, and $D_{i}\subseteq\Delta(0,R)$.
The other possible output is $Failed$.
\begin{enumerate}
\item If $0\in Z_{i}$ for all $0\leq i\leq N$ return $Failed$. Otherwise
let $d$ be the maximal $i$ such that $0\notin Z_{i}$.
\item Put \[
R:=\min{}_{d<i\leq N}(\frac{\underline{Z_{d}}}{N\overline{Z_{i}}})^{\frac{1}{i-d}}\]
If \[
R\leq\max{}_{0\leq i<d}(\frac{N\overline{Z_{i}}}{\underline{Z_{d}}})^{\frac{1}{d-i}}\]
return $Failed$.
\item Set $f_{c}=b_{d}z^{d}+\ldots+b_{0}$, where $b_{i}=\gamma(Z_{i})$
for $0\leq i\leq d$, and set \[p=-\max_{0\leq i\leq d}\log\rho(Z_{i})\]
\item Compute $(s_{1},\ldots,s_{d})=ApproximateRoots(f_{c},p)$.
\item Let $F=Z_{N}z^{N}+\ldots+Z_{0}$. For each $1\leq j\leq d$ and $0\leq i\leq N$
use complex disc arithmetic to compute $W_{i,j}:=\frac{F^{(i)}(s_{j})}{i!}$.
\item For each $1\leq j\leq d$ let $\tilde{k_{j}}$ be the smallest $k>0$
such that $\underline{W_{k,j}}\neq0$ and\[
r_{j}:=\max{}_{0\leq i<k}(\frac{N\overline{W_{i,j}}}{\underline{W_{k,j}}})^{\frac{1}{k-i}}<\min{}_{k<i\leq N}(\frac{\underline{W_{k,j}}}{N\overline{W_{i,j}}})^{\frac{1}{i-k}}\]
If there is no $k$ satisfying the condition return $Failed$.
\item Find the connected components of the union of $\triangle(s_{j},r_{j})$.
Let $J_{1},\ldots,J_{m}$ be the sets of indices $j$ corresponding
to the connected components. 
\item For each $1\leq l\leq m$, let $J_{l}=\{j_{l,1},\ldots j_{l,k_{l}}\}$.
If for some $j\in J_{l}$ $\tilde{k_{j}}\neq k_{l}$ return $Failed$.
Otherwise pick $j\in J_{l}$ with the minimal value of $r_{j}$, and
set $D_{l}:=\Delta(s_{j},r_{j})$. 
\item If $D_{l}\nsubseteq\Delta(0,R)$ for some $1\leq l\leq m$ return
$Failed$.
\item Return $(D_{1},\ldots,D_{m})$, $(k_{1},\ldots,k_{m})$, and $R$.
\end{enumerate}
\end{algorithm}
To show correctness of Algorithm \ref{IntervalRoots}, suppose that
the algorithm returned $(D_{1},\ldots,D_{m})$, $(k_{1},\ldots,k_{m})$,
and $R$. Let $a_{0}\in Z_{0},\ldots,a_{N}\in Z_{N}$ and $f=a_{N}z^{N}+\ldots+a_{0}$.
By Proposition \ref{RootMultProp}, and because the algorithm did
not fail in step $(2)$, $f$ has exactly $d$ roots in $\Delta(0,R)$.
The condition in step $(6)$ and Proposition \ref{RootMultProp} imply
that $f$ has exactly $k_{i}$ roots in the disc $D_{i}$. Since the
algorithm did not fail in step $(8)$, $k_{1}+\ldots+k_{m}=d$, and
hence $f$ has no roots in $\Delta(0,R)\setminus\bigcup_{i=1}^{m}D_{i}$.
Step $(7)$ guarantees that for any $1\leq i<j\leq m$, $D_{i}\cap D_{j}=\emptyset$.
Step $(9)$ ensures that $D_{i}\subseteq\Delta(0,R/2)$. 

Computation of sample points in CAD cells requires a representation
of vectors with algebraic number coordinates. The following gives
a recursive definition of root isolation data and of representation
of real algebraic vectors. Note that root isolation data provides
information about roots of $f_{k}$ not only at $u$, but also in
a neighbourhood of $u$ (point $(5)$ of the definition). This property
is crucial for computing cell adjacencies.
\begin{defn}
\label{RAV}\emph{$\Theta_{k}=((D_{1},\ldots,D_{m}),(k_{1},\ldots,k_{m}),R)$}
is \emph{root isolation data} for $f_{k}\in\mathbb{Q}[x_{1},\ldots,x_{k}]$
at $u=(\alpha_{1},\ldots,\alpha_{k-1})\in\mathbb{R}^{k-1}$ if
\begin{enumerate}
\item $D_{1},\ldots,D_{m}\in I_{2}(\mathbb{C})$, $k_{1},\ldots,k_{m}\in\mathbb{N}_{+}$,
and $R>0$,
\item $g_{k}=f_{k}(\alpha_{1},\ldots,\alpha_{k-1},x_{k})$ has a root of
multiplicity $k_{j}$ in $D_{j}$, for $1\leq j\leq m$, and has no
other roots,
\item for any $1\leq j\leq m$ $dbl(D_{j})\subseteq\Delta(0,R/2)$,
\item for any $j_{1}\neq j_{2}$ $dbl(D_{j_{1}})\cap dbl(D_{j_{2}})=\emptyset$
and if $dbl(D_{j_{1}})\cap\mathbb{R}\neq\emptyset$ then $conj(dbl(D_{j_{1}}))\cap dbl(D_{j_{2}})=\emptyset$,
\item if $W_{i}$ is the isolating disc of $\alpha_{i}$ for $1\leq i<k$,
and $\beta_{i}\in quad(W_{i})$ then $f_{k}(\beta_{1},\ldots,\beta_{k-1},x_{k})$
has exactly $k_{j}$ roots in $D_{j}$, multiplicities counted, and
has no roots in $\Delta(0,R)\setminus\bigcup_{j=1}^{m}D_{j}$.
\end{enumerate}
A \emph{real algebraic vector} $v=RAV(\Omega_{k})=(\alpha_{1},\ldots,\alpha_{k})\in\mathbb{R}^{k}$
is represented by \[\Omega_{k}=(\Omega_{k-1},f_{k},W_{k},\Theta_{k})\]
where
\begin{itemize}
\item $\Omega_{k-1}$ represents the real algebraic vector $\Pi(v)=RAV(\Omega_{k-1})=(\alpha_{1},\ldots,\alpha_{k-1})\in\mathbb{R}^{k-1}$, 
\item $f_{k}\in\mathbb{Q}[x_{1},\ldots,x_{k}]$ is a \emph{defining polynomial}
of $\alpha_{k}$,
\item $W_{k}\in I_{2}(\mathbb{C})$ is an \emph{isolating disc} of $\alpha_{k}$,
\item \emph{$\Theta_{k}=((D_{1},\ldots,D_{m}),(k_{1},\ldots,k_{m}),R)$}
is\emph{ }root isolation data\emph{ }for $f_{k}$\emph{ at $\Pi(v)$},
$W_{k}=D_{j_{0}}$ for some $1\leq j_{0}\leq m$, and $W_{k}\cap\mathbb{R}\neq\emptyset$.
\end{itemize}
Define $\rho(\Omega_{k})=\max(\rho(\Omega_{k-1}),\max_{1\leq j\leq m}\rho(D_{j}))$.

To complete the recursive definition let $\Omega_{0}=()$ be the representation
of the only element of $\mathbb{R}^{0}$. 

We will say that $\Omega_{k}^{\prime}=(\Omega_{k-1}^{\prime},f_{k},W_{k}^{\prime},\Theta_{k}^{\prime})$
is a \emph{refinement} of $\Omega_{k}$ if $RAV(\Omega_{k}^{\prime})=RAV(\Omega_{k})$,\emph{
$\Theta_{k}^{\prime}=((D_{1}^{\prime},\ldots,D_{m}^{\prime}),(k_{1},\ldots,k_{m}),R^{\prime})$,}
\emph{$\rho(D_{j}^{\prime})<\rho(D_{j})$} and $D_{j}^{\prime}\subseteq dbl(D_{j})$
for $1\leq j\leq m$, and $\Omega_{k-1}^{\prime}$ is a refinement
of $\Omega_{k-1}$.

For any $v=(a_{1},\ldots,a_{k})\in\mathbb{R}^{k}$ and $a\in\mathbb{R}$,
we will use notation $\Lambda(v)=a_{k}$ and $v\times a=(a_{1},\ldots,a_{k},a)\in\mathbb{R}^{k+1}$.\end{defn}
\begin{rem}
\label{RefRem}A refinement $\Omega_{k}^{\prime\prime}$ of a refinement
$\Omega_{k}^{\prime}$ of $\Omega_{k}$ is a refinement of $\Omega_{k}$.\end{rem}
\begin{proof}
With notations from Definition \ref{RAV}, let \[
\Omega_{k}^{\prime\prime}=(\Omega_{k-1}^{\prime\prime},f_{k},W_{k}^{\prime\prime},\Theta_{k}^{\prime\prime})\]
and\emph{ }\[
\Theta_{k}^{\prime\prime}=((D_{1}^{\prime\prime},\ldots,D_{m}^{\prime\prime}),(k_{1},\ldots,k_{m}),R^{\prime\prime})\]
By induction, it suffices to show \emph{$\rho(D_{j}^{\prime\prime})<\rho(D_{j})$}
and $D_{j}^{\prime\prime}\subseteq dbl(D_{j})$ for $1\leq j\leq m$.
$\rho(D_{j}^{\prime\prime})<\rho(D_{j})$ follows from \emph{$\rho(D_{j}^{\prime\prime})<\rho(D_{j}^{\prime})$
and} $\rho(D_{j}^{\prime})<\rho(D_{j})$.\emph{ }$D_{j}^{\prime\prime}\subseteq dbl(D_{j}^{\prime})$
and $D_{j}^{\prime}\subseteq dbl(D_{j})$ implies that $D_{j}$, $D_{j}^{\prime}$
and $D_{j}^{\prime\prime}$ contain the same root of $f_{k}(\alpha_{1},\ldots,\alpha_{k-1},x_{k})$.
Then \emph{$\rho(D_{j}^{\prime\prime})<\rho(D_{j})$ implies }$D_{j}^{\prime\prime}\subseteq dbl(D_{j})$. 
\end{proof}
The algorithms we introduce next take a \emph{working precision} argument.
A working precision $p$ is a positive integer. One can think of it
as the number of bits in floating-point numbers used in a numeric
approximation algorithm. However, we will not attach any specific
meaning to the working precision argument. Instead our algorithms
will satisfy certain properties as $p$ tends to infinity. For instance,
if we say that a certain quantity $\omega$ in the output of an algorithm
tends to zero as $p$ tends to infinity, it means that for any $\epsilon>0$
there exists $N>0$ such that if the working precision $p>N$ then
the algorithm will produce an output with $\omega<\epsilon$.

Let $v=(\alpha_{1},\ldots,\alpha_{k})$ be a real algebraic vector
and let \[
f\in\mathbb{Q}[x_{1},\ldots,x_{k},x_{k+1}]\]
be such that $f(\alpha_{1},\ldots,\alpha_{k},x_{k+1})$ does not vanish
identically. We will now describe an algorithm $AlgRoots_{k}$, with
$k\geq0$, which finds the root isolation data of $f$ at $v$ and
the real roots of $f(\alpha_{1},\ldots,\alpha_{k},x_{k+1})$. The
algorithm uses two subalgorithms $Refine_{k}$ and $ZeroTest_{k}$
that will be defined recursively in terms of $AlgRoots_{k-1}$. Given
\[
(\alpha_{1},\ldots,\alpha_{k})=RAV(\Omega_{k})\]
and a working precision $p>0$ $Refine_{k}$ computes a refinement
$\Omega_{k}^{\prime}$ of $\Omega_{k}$ such that as $p$ tends to
infinity $\rho(\Omega_{k}^{\prime})$ tends to zero. $ZeroTest_{k}$
decides whether $h(\alpha_{1},\ldots,\alpha_{k})$ is zero for a given
$h\in\mathbb{Q}[x_{1},\ldots,x_{k}]$.
\begin{algorithm}
\label{AlgRoots}($AlgRoots_{k}$)\\
Input: Real algebraic vector $v=(\alpha_{1},\ldots,\alpha_{k})=RAV(\Omega_{k})$,
where $k\geq0$, $f\in\mathbb{Q}[x_{1},\ldots,x_{k},x_{k+1}]$, such
that $f(\alpha_{1},\ldots,\alpha_{k},x_{k+1})$ does not vanish identically,
and a working precision $p>0$. \\
Output: Root isolation data $\Theta$ of $f$ at $v$, a refinement
$\Omega_{k}^{\prime}$ of $\Omega_{k}$, and real algebraic vectors
$v_{1}=RAV(\Omega_{k+1,1}),\ldots,v_{r}=RAV(\Omega_{k+1,r})\in\mathbb{R}^{k+1}$
such that $\Pi(v_{j})=RAV(\Omega_{k}^{\prime})$, for $1\leq j\leq r$,
$\Lambda(v_{1}),\ldots,\Lambda(v_{r})$ are all the real roots of
$f(\alpha_{1},\ldots,\alpha_{k},x_{k+1})$, and as $p$ tends to infinity
$\rho(\Omega_{k+1,j})$ tends to zero.
\begin{enumerate}
\item Let $f=a_{N}x_{k+1}^{N}+\ldots+a_{0}$. Find $d$ such that $a_{d}(\alpha_{1},\ldots,\alpha_{k})\neq0$
and \[
a_{d+1}(\alpha_{1},\ldots,\alpha_{k})=\ldots=a_{N}(\alpha_{1},\ldots,\alpha_{k})=0\]
using $ZeroTest_{k}$ if $k>0$. Set $g=a_{d}z^{d}+\ldots+a_{0}$.
\item Compute the principal subresultant coefficients $psc_{0},\ldots,psc_{n-1}$
of $g$ and $\frac{\partial g}{\partial z}$ with respect to $z$.
\item Find the largest integer $\mu\geq0$ such that \[
psc_{0}(\alpha_{1},\ldots,\alpha_{k})=\ldots=psc_{\mu-1}(\alpha_{1},\ldots,\alpha_{k})=0\]
using $ZeroTest_{k}$ if $k>0$.
\item Set $p^{\prime}=p$ and $\Omega_{k}^{\prime}=\Omega_{k}$.
\item If $k>0$ compute $\Omega_{k}^{\prime}=Refine_{k}(\Omega_{k}^{\prime},p^{\prime})$
and let $W_{1},\ldots,W_{k}$ be the isolating discs of $\alpha_{1},\ldots,\alpha_{k}$
in $\Omega_{k}^{\prime}$.
\item For $0\leq i\leq N$

\begin{enumerate}
\item if $a_{i}\in\mathbb{Q}$ compute $Z_{i}\in I_{2}(\mathbb{C})$ such
that $a_{i}\in Z_{i}$ and $\rho(a_{i})\leq2^{-p^{\prime}}$,
\item else compute $Z_{i}=a_{i}(quad(W_{1}),\ldots,quad(W_{k}))$ using
interval arithmetic.
\end{enumerate}
\item Compute $\Theta:=IntervalRoots(Z_{0},\ldots,Z_{N})$. 
\item If $\Theta=Failed$ double $p^{\prime}$ and go to step $(5)$. 
\item Let $\Theta=((D_{1},\ldots,D_{m}),(k_{1},\ldots,k_{m}),R)$. If $d-m>\mu$
or the conditions $(3)$ and $(4)$ of Definition \ref{RAV} are not
satisfied, double $p^{\prime}$ and go to step $(5)$. 
\item Let $(j_{1},\ldots,j_{r})$ be the set of indices for which $D_{j_{l}}\cap\mathbb{R}\neq\emptyset$.
For $1\leq l\leq r$ let $\Omega_{k+1,l}=(\Omega_{k}^{\prime},f,D_{j_{l}},\Theta)$
and $v_{l}=RAV(\Omega_{k+1,l})$.
\item Return $\Theta$, $\Omega_{k}^{\prime}$, and $v_{1},\ldots,v_{r}$.
\end{enumerate}
\end{algorithm}
\begin{proof}
To prove termination of the Algorithm \ref{AlgRoots} we need to show
that for sufficiently large $p^{\prime}$ the call to \emph{IntervalRoots}
in step $(7)$ succeeds and gives a result with $d-m=\mu$. The specification
of the algorithm $Refine_{k}$ implies that as $p^{\prime}$ tends
to infinity $\max_{1\leq i\leq k}\rho(W_{i})$ tends to zero. Hence
also $\max_{1\leq i\leq N}\rho(Z_{i})$ tends to zero. Therefore,
as $p^{\prime}$ tends to infinity, $\gamma(Z_{i})$ tends to $a_{i}(\alpha_{1},\ldots,\alpha_{k})$
and $\underline{Z_{i}}$ and $\overline{Z_{i}}$ tend to $|a_{i}(\alpha_{1},\ldots,\alpha_{k})|$
for all $0\leq i\leq N$. In particular, for sufficiently large $p^{\prime}$,
$0\notin Z_{i}$ iff $a_{i}(\alpha_{1},\ldots,\alpha_{k})\neq0$,
and hence $d$ in step $(1)$ of \emph{IntervalRoots} is the same
as $d$ computed in step $(1)$ of Algorithm \ref{AlgRoots}. Since\[
R:=\min{}_{d<i\leq N}(\frac{\underline{Z_{d}}}{N\overline{Z_{i}}})^{\frac{1}{i-d}}\]
tends to infinity and \[
\max{}_{0\leq i<d}(\frac{N\overline{Z_{i}}}{\underline{Z_{d}}})^{\frac{1}{d-i}}\]
tends to a finite constant, the call to \emph{IntervalRoots} does
not fail in in step $(2)$ for sufficiently large $p^{\prime}$. Let
$\sigma_{1},\ldots\sigma_{d}$ be the roots of $h(z)=g(\alpha_{1},\ldots,\alpha_{k},z)$,
each repeated as many times as its multiplicity. Let $s_{1},\ldots,s_{d}$
be the roots computed in step $(4)$ of \emph{IntervalRoots}. The
specification of \emph{ApproximateRoots} implies that, as $p^{\prime}$
tends to infinity, after a suitable reordering of roots, $s_{j}$
tends to $\sigma_{j}$ for each $1\leq j\leq d$. Hence for $W_{i,j}:=\frac{F^{(i)}(s_{j})}{i!}$
computed in step $(5)$ of \emph{IntervalRoots} $\gamma(W_{i,j})$
tends to $\frac{h^{(i)}(\sigma_{j})}{i!}$ and $\rho(W_{i,j})$ tends
to zero. Therefore, $\underline{W_{i,j}}$ and $\overline{W_{i,j}}$
tend to $|\frac{h^{(i)}(\sigma_{j})}{i!}|$. Hence, if $k_{j}$ is
the multiplicity of $\sigma_{j}$, \[
r_{j}:=\max{}_{0\leq i<k_{j}}(\frac{N\overline{W_{i,j}}}{\underline{W_{k_{j},j}}})^{\frac{1}{k_{j}-i}}\]
tends to zero and\[
\min{}_{k_{j}<i\leq N}(\frac{\underline{W_{k_{j},j}}}{N\overline{W_{i,j}}})^{\frac{1}{i-k_{j}}}\]
is bounded away from zero for sufficiently large $p^{\prime}$. Therefore,
for sufficiently large $p^{\prime}$, the condition in step $(6)$
of \emph{IntervalRoots} is satisfied by $k_{j}$. Note that if $s_{j}$
is closer to $\sigma_{j}$ than to other roots of $h$ then the condition
cannot be satisfied by any $k<k_{j}$. Otherwise, by Proposition \ref{RootMultProp},
$h$ would have exactly $k$ roots in $\Delta(s_{j},r_{j})$ which
is impossible since if $\sigma_{j}\notin\Delta(s_{j},r_{j})$ then
$\Delta(s_{j},r_{j})$ contains no roots of $h$ and else $\Delta(s_{j},r_{j})$
contains at least $k_{j}$ roots of $h$. Hence, for sufficiently
large $p^{\prime}$, step $(6)$ does not fail and $\tilde{k_{j}}=k_{j}$
for each$1\leq j\leq d$. Since $r_{j}$ tends to zero for $1\leq j\leq d$,
for sufficiently large $p^{\prime}$, $\Delta(s_{j_{1}},r_{j_{1}})$
and $\Delta(s_{j_{2}},r_{j_{2}})$ intersect iff $\sigma_{j_{1}}=\sigma_{j_{2}}$,
and hence step $(8)$ of \emph{IntervalRoots} does not fail. Since
$R$ tends to infinity, for sufficiently large $p^{\prime}$, step
$(9)$ of \emph{IntervalRoots} does not fail and the whole algorithm
succeeds. Since, for sufficiently large $p^{\prime}$, the discs returned
by \emph{IntervalRoots} correspond to distinct roots of $h$, $d-m=\mu$
in step $(9)$ of Algorithm \ref{AlgRoots}. Since $\rho(D_{i})$
tends to zero for $1\leq i\leq m$ and $R$ tends to infinity, for
sufficiently large $p^{\prime}$, we have $dbl(D_{i})\subseteq\Delta(0,R/2)$
for any $1\leq i\leq m$ and

\[
\rho(D_{i})<\frac{\min_{\sigma_{j_{1}}\neq\sigma_{j_{2}}}|\sigma_{j_{1}}-\sigma_{j_{2}}|}{16}\]
Therefore for any $i_{1}\neq i_{2}$ $dbl(D_{i_{1}})\cap dbl(D_{i_{2}})=\emptyset$
and if $\sigma$ is the root of $h$ in $D_{i_{1}}$ then either $\sigma\notin\mathbb{R}$
and $dbl(D_{i_{1}})\cap\mathbb{R}=\emptyset$ or $\sigma\in\mathbb{R}$
and $conj(dbl(D_{i_{1}}))\cap dbl(D_{i_{2}})=\emptyset$. Hence for
sufficiently large $p^{\prime}$, the conditions $(3)$ and $(4)$
of Definition \ref{RAV} are satisfied and the algorithm terminates.

Proposition 4.3 of \cite{S7} and correctness of Algorithm \ref{IntervalRoots}
imply that $\Omega_{k+1,l}$ satisfy the conditions $(1)$, $(2)$,
and $(5)$ of Definition \ref{RAV}, and the conditions $(3)$ and
$(4)$ are ensured by step $(9)$ of Algorithm \ref{AlgRoots}. Step
$(10)$ selects all isolating discs that intersect the real line,
hence $\Lambda(v_{1}),\ldots,\Lambda(v_{r})$ are all the real roots
of $f(\alpha_{1},\ldots,\alpha_{k},x_{k+1})$.

To show that $\rho(\Omega_{k+1,j})$ tends to zero as $p$ tends to
infinity, note that $p^{\prime}\geq p$, we have already shown that
$\rho(D_{i})$ tends to zero for $1\leq i\leq m$ as $p^{\prime}$
tends to infinity, and since $\Omega_{k}^{\prime}$ is computed by
$Refine_{k}$ with working precision $p^{\prime}$, $\rho(\Omega_{k}^{\prime})$
tends to zero as $p^{\prime}$ tends to infinity. \end{proof}
\begin{rem}
If $m=d$ in step $(10)$ of Algorithm \ref{AlgRoots} then \[
g(\alpha_{1},\ldots,\alpha_{k},z)\]
does not have multiple roots and computing the principal subresultant
coefficients is not necessary. Hence instead of computing the principal
subresultant coefficients in step $(2)$ it is sufficient to compute
them only when the algorithm reaches step $(9)$ for the first time
and $m<d$. 
\end{rem}
To complete the description of Algorithm \ref{AlgRoots} let us now
define the subalgorithms $Refine_{k}$ and $ZeroTest_{k}$.
\begin{algorithm}
\label{Refine}($Refine_{k}$)\\
Input: Real algebraic vector $(\alpha_{1},\ldots,\alpha_{k})=RAV(\Omega_{k})$,
where $k\geq1$, and a working precision $p>0$. \\
Output: A refinement $\Omega_{k}^{\prime}$ of $\Omega_{k}$ such
that as $p$ tends to infinity \textup{$\rho(\Omega_{k}^{\prime})$}
tends to zero.
\begin{enumerate}
\item Let $\Omega_{k}=(\Omega_{k-1},f_{k},W_{k},\Theta_{k})$, \[
\Theta_{k}=((D_{1},\ldots,D_{m}),(k_{1},\ldots,k_{m}),R)\]
and $W_{k}=D_{j_{0}}$. Set $v=RAV(\Omega_{k-1})$ and $p^{\prime}=p$.
\item Compute \[
(\Theta;\Omega_{k-1}^{\prime};v_{1},\ldots,v_{r})=AlgRoots_{k-1}(v,f_{k},p^{\prime})\]
where \emph{$\Theta=((D_{1}^{\prime},\ldots,D_{m}^{\prime}),(k_{1},\ldots,k_{m}),R^{\prime})$.}
\item \emph{If no reordering of indices in $\Theta$ yields} \emph{$\rho(D_{j}^{\prime})<\rho(D_{j})$}
and $D_{j}^{\prime}\subseteq dbl(D_{j})$ for $1\leq j\leq m$, double
$p^{\prime}$ and go to $(2)$.
\item Let $v_{j}=RAV(\Omega_{k}^{\prime})$ be such that the isolating disk
of $\Lambda(v_{j})$ is $D_{j_{0}}^{\prime}$. Return $\Omega_{k}^{\prime}$. 
\end{enumerate}
\end{algorithm}
Since $\rho(D_{j}^{\prime})$ tends to zero as $p^{\prime}$ tends
to infinity, for sufficiently large $p^{\prime}$ the condition in
step $(3)$ is satisfied for pairs $D_{j}$ and $D_{j}^{\prime}$
containing the same root of $f_{k}(\alpha_{1},\ldots,\alpha_{k-1},x_{k})$,
and hence the algorithm terminates. Correctness of Algorithm \ref{AlgRoots}
and the condition in step $(3)$ guarantee that $\Omega_{k}^{\prime}$
is a refinement of $\Omega_{k}$ and as $p$ tends to infinity $\rho(\Omega_{k}^{\prime})$
tends to zero.
\begin{algorithm}
\label{ZeroTest}($ZeroTest_{k}$)\\
Input: Real algebraic vector $(\alpha_{1},\ldots,\alpha_{k})=RAV(\Omega_{k})$,
where $k\geq1$, and $h\in\mathbb{Q}[x_{1},\ldots,x_{k}]$. \\
Output: $true$ if $h(\alpha_{1},\ldots,\alpha_{k})=0$ and $false$
otherwise. 
\begin{enumerate}
\item Let $\Omega_{k}=(\Omega_{k-1},f_{k},W_{k},\Theta_{k})$, \[
\Theta_{k}=((D_{1},\ldots,D_{m}),(k_{1},\ldots,k_{m}),R)\]
and $W_{k}=D_{j_{0}}$. Set $\mu=k_{j_{0}}$, $\Omega_{k}^{\prime}=\Omega_{k}$,
and set an initial value $p$ of working precision (e.g. to precision
that was used to compute $\Omega_{k}$).
\item Compute $\Omega_{k}^{\prime}=Refine_{k}(\Omega_{k}^{\prime},p)$.
Let $\Omega_{k}^{\prime}=(\Omega_{k-1}^{\prime},f_{k},W_{k}^{\prime},\Theta_{k}^{\prime})$.
Set $v=RAV(\Omega_{k-1}^{\prime})$.
\item Compute $(\Theta;\Omega_{k-1}^{\prime};v_{1},\ldots,v_{r})=AlgRoots_{k-1}(v,f_{k}h,p)$.
\item If $W_{k}^{\prime}$ intersects the isolating disc of $\Lambda(v_{j})$
for more than one $j$, double $p$ and go to step $(2)$.
\item Let $j$ be the only index for which $W_{k}^{\prime}$ intersects
the isolating disc $W$ of $\Lambda(v_{j})$. Let \emph{$\Theta=((\tilde{D_{1}},\ldots,\tilde{D_{m}}),(\tilde{k_{1}},\ldots,\tilde{k_{m}}),\tilde{R})$},
and $W=\tilde{D_{\tilde{j_{0}}}}$.
\item If $\tilde{k_{\tilde{j_{0}}}}>\mu$ return $true$ otherwise return
$false$.
\end{enumerate}
\end{algorithm}
Since as $p^{\prime}$ tends to infinity, $\rho(W_{k}^{\prime})$
and $\max_{1\leq j\leq r}\rho(\Lambda(v_{j}))$ tend to zero, for
sufficiently large $p^{\prime}$, $W_{k}^{\prime}$ intersects only
the isolating disc of $\Lambda(v_{j})=\alpha_{k}$, which proves termination
and correctness of the algorithm. 

Since $AlgRoots_{k}$ is defined for $k\geq0$ and $Refine_{k}$ and
$ZeroTest_{k}$ are defined for $k\geq1$, the recursive definition
of the algorithms is complete. 
\begin{rem}
$ZeroTest_{k}$ is defined here in terms of $AlgRoots_{k-1}$ for
simplicity of description. In practice, to decide whether \[
h(\alpha_{1},\ldots,\alpha_{k})=0\]
we can first evaluate $h$ at the isolating discs of $\alpha_{1}\ldots,\alpha_{k}$
using interval arithmetic. If the result does not contain zero then
\[
h(\alpha_{1},\ldots,\alpha_{k})\neq0\]
Otherwise, we isolate roots of \[
h_{\alpha}=h(\alpha_{1},\ldots,\alpha_{k-1},z)\]
and refine isolating discs of roots of \[
g_{\alpha}=g(\alpha_{1},\ldots,\alpha_{k-1},z)\]
and roots of $h_{\alpha}$ until either the isolating disc of $\alpha_{k}$
does not intersect any isolating discs of roots of $h_{\alpha}$ or
the number of intersecting isolating discs of roots of $g_{\alpha}$
and $h_{\alpha}$ agrees with the number of common roots of $g_{\alpha}$
and $h_{\alpha}$ computed by finding signs of principal subresultant
coefficients of $g_{\alpha}$ and $h_{\alpha}$ (see Proposition 4.4
of \cite{S7}). When the algorithm is used in CAD construction we
also use information about polynomials that are zero at the current
point that was collected during the construction (see \cite{S7},
Section 4.1).
\end{rem}

\section{Finding cell adjacencies}

Let $F$ be a well-based CAF in $x_{1},\ldots,x_{n}$. For simplicity
let us assume that $BP(F)=\{f_{1},\ldots,f_{n}\}$, where \[
f_{k}\in\mathbb{\mathbb{Q}}[x_{1},\ldots,x_{k}]\setminus\mathbb{\mathbb{Q}}[x_{1},\ldots,x_{k-1}]\]
This can be always achieved by multiplying all elements of \[
BP(F)\cap(\mathbb{\mathbb{Q}}[x_{1},\ldots,x_{k}]\setminus\mathbb{\mathbb{Q}}[x_{1},\ldots,x_{k-1}])\]
We can also assume that $f_{1},\ldots,f_{n}$ are square-free. 

The main algorithm \emph{CADAdjacency} (Algorithm \ref{CADAdjacency})
finds all pairs of adjacent cells $(C,C^{\prime})\in D(F)^{2}$ such
that $\dim C-\dim C^{\prime}=1$. By Proposition \ref{WellBasedProp},
to determine all cell adjacencies for a well-based CAF it is sufficient
to find all pairs of adjacent cells whose dimensions differ by one,
hence Algorithm \ref{CADAdjacency} is sufficient to fully solve the
cell adjacency problem for well-based CAF.

The algorithm first calls \emph{SamplePoints} (Algorithm \ref{SPT}),
which constructs a sample point $SPT(C)=(a_{1},\ldots,a_{k})\in\mathbb{R}^{k}$
in each cell $C\in D(\Pi_{k}(F))$, for $1\leq k\leq n$, and computes
root isolation data $RTS(C)$ for each cell $C\in D(\Pi_{k}(F))$,
for $1\leq k<n$. Let us describe the representation of sample points
and give the specification of root isolation data. Let $I=(i_{1},\ldots,i_{l})$
be the set of indices $1\leq i\leq k$ such that $\Pi_{i}(C)$ is
a section. For $i\notin I$, $a_{i}$ is a rational number and for
$i\in I$, $a_{i}$ is an algebraic number. To represent sample points
we will use combinations of rational vectors and algebraic vectors
defined as follows. Let $1\leq k\leq n$, let $0\leq l\leq k$, let
$I=\{i_{1},\ldots,i_{l}\}$, where $1\leq i_{1}<\ldots<i_{l}\leq k$,
and let $J=\{1,\ldots,k\}\setminus I=\{j_{1},\ldots,j_{k-l}\}$, where
$1\leq j_{1}<\ldots<j_{k-l}\leq k$. Let $v=(\alpha_{1},\ldots,\alpha_{l})=RAV(\Omega_{l})$
be a real algebraic vector and let $w=(q_{1},\ldots,q_{k-l})\in\mathbb{Q}^{k-l}$
be a rational vector. By $PT(v,w,I,J)$ we denote the point $(a_{1},\ldots,a_{k})\in\mathbb{R}^{k}$
such that $a_{i_{s}}=\alpha_{s}$ for $1\leq s\leq l$ and $a_{j_{t}}=q_{t}$
for $1\leq t\leq k-l$. Suppose $1\leq k<n$ and $SPT(C)=PT(v,w,I,J)$.
Let $f_{k+1}^{C}\in\mathbb{Q}[x_{i_{1}},\ldots,x_{i_{l}},x_{k+1}]$
denote $f_{k+1}$ with $x_{j_{t}}$ replaced by $a_{j_{t}}$ for $1\leq t\leq k-l$.
Then $RTS(C)$ computed by \emph{SamplePoints} is root isolation data
of $f_{k+1}^{C}$ at $v$.

Next \emph{CADAdjacency} calls \emph{AdjacencyPoints} (Algorithm \ref{ADP})
which, for $1\leq k\leq n$, and for each pair of adjacent cells $(C,C^{\prime})$
of $D(\Pi_{k}(F))$ with $\dim C^{\prime}=\dim C-1$, constructs a
point $ADP(C,C^{\prime})\in C$ which satisfies the following condition.
\begin{condition}
\label{ADPCond}If $SPT(C^{\prime})=(a_{1},\ldots,a_{k})$ and $ADP(C,C^{\prime})=(b_{1},\ldots,b_{k})$
then 
\begin{itemize}
\item for each $1\leq i\leq k$ if $a_{i}$ is a root of $f_{i}(a_{1},\ldots,a_{i-1},x_{i})$
with isolating disc $W_{i}$ then $b_{i}\in dbl(W_{i})$, 
\item if $a_{i}$ is a rational number between roots of $f_{i}(a_{1},\ldots,a_{i-1},x_{i})$
then $b_{i}=a_{i}$. 
\end{itemize}
\end{condition}
Finally, \emph{CADAdjacency} returns the pairs of cells $(C,C^{\prime})\in D(F)^{2}$
for which $ADP(C,C^{\prime})$ is defined.
\begin{algorithm}
\label{CADAdjacency}(CADAdjacency)\\
Input: A well-based CAF $F$ in $x_{1},\ldots,x_{n}$ with $BP(F)=\{f_{1},\ldots,f_{n}\}$,
where $f_{k}\in\mathbb{\mathbb{Q}}[x_{1},\ldots,x_{k}]\setminus\mathbb{\mathbb{Q}}[x_{1},\ldots,x_{k-1}]$.
\\
Output: The set $A$ of all pairs of adjacent cells $(C,C^{\prime})\in D(F)^{2}$
such that $\dim C-\dim C^{\prime}=1$.
\begin{enumerate}
\item Compute $(SPT,RTS)=SamplePoints(F)$. 
\item Compute $ADP=AdjacencyPoints(F,SPT,RTS)$.
\item Return the set of all pairs of cells $(C,C^{\prime})\in D(F)^{2}$
such that $ADP(C,C^{\prime})$ is defined.
\end{enumerate}
\end{algorithm}

\begin{algorithm}
\label{SPT}(SamplePoints)\\
Input: A well-based CAF $F$ in $x_{1},\ldots,x_{n}$ with $BP(F)=\{f_{1},\ldots,f_{n}\}$,
where $f_{k}\in\mathbb{\mathbb{Q}}[x_{1},\ldots,x_{k}]\setminus\mathbb{\mathbb{Q}}[x_{1},\ldots,x_{k-1}]$.
\\
Output: $SPT$ and $RTS$ such that
\begin{itemize}
\item for $1\leq k\leq n$ and for each cell $C$ of $D(\Pi_{k}(F))$, $SPT(C)$
is a sample point in $C$,
\item for $1\leq k<n$ and for each cell $C$ of $D(\Pi_{k}(F))$, $RTS(C)$
is root isolation data for $C$.\end{itemize}
\begin{enumerate}
\item Set an initial value $p$ of working precision.
\item Compute $(\Theta;();v_{1},\ldots,v_{r})=AlgRoots_{0}((),f_{1},p)$.
We have \[\Theta=((D_{1},\ldots,D_{m}),(k_{1},\ldots,k_{m}),R)\]
\item Pick rational numbers $-R<q_{1}<\Lambda(v_{1})<q_{2}<\ldots<q_{r}<\Lambda(v_{r})<q_{r+1}<R$
such that $q_{i}\notin\bigcup_{j=1}^{m}dbl(D_{j})$.
\item For $1\leq i\leq r$, let $C$ be the $i$-th $\{f_{1}\}$-section.
Set $SPT(C)=PT(v_{i},(),\{1\},\{\})$.
\item For $1\leq i\leq r+1$, let $C$ be the $i$-th $\{f_{1}\}$-sector.
Set $SPT(C)=PT((),(q_{i}),\{\},\{1\})$.
\item For $1\leq k<n$ and for each cell $C$ of $D(\Pi_{k}(F))$:

\begin{enumerate}
\item Let $(a_{1},\ldots,a_{k})=PT(v,w,I,J)=SPT(C)$.
\item Let $I=\{i_{1},\ldots,i_{l}\}$, $J=\{j_{1},\ldots,j_{k-l}\}$, and
let $f_{k+1}^{C}$ be $f_{k+1}$ with $x_{j_{t}}$ replaced by $a_{j_{t}}$
for $1\leq t\leq k-l$. 
\item Compute $(\Theta;\Omega_{l}^{\prime};v_{1},\ldots,v_{r})=AlgRoots_{l}(v,f_{k+1}^{C},p)$.
We have \[\Theta=((D_{1},\ldots,D_{m}),(k_{1},\ldots,k_{m}),R)\]
\item Set $RTS(C)=\Theta$ and replace representations of algebraic vectors
in $SPT(\Pi_{i}(C))$ for $i\leq k$ with their refinements that appear
in $\Omega_{l}^{\prime}$.
\item Pick rational numbers $-R<q_{1}<\Lambda(v_{1})<q_{2}<\ldots<q_{r}<\Lambda(v_{r})<q_{r+1}<R$
such that $q_{i}\notin\bigcup_{j=1}^{m}dbl(D_{j})$, and let $w_{i}=w\times q_{i}$
for $1\leq i\leq r+1$.
\item For $1\leq i\leq r$, let $S$ be the $i$-th $\{f_{k+1}\}$-section
over $C$. Set \[SPT(S)=PT(v_{i},w,I\cup\{k+1\},J)\]
\item For $1\leq i\leq r+1$, let $S$ be the $i$-th $\{f_{k+1}\}$-sector
over $C$. Set \[SPT(S)=PT(v,w_{i},I,J\cup\{k+1\})\]
\end{enumerate}
\item Return $SPT$ and $RTS$.
\end{enumerate}
\end{algorithm}

\begin{algorithm}
\label{ADP}(AdjacencyPoints)\\
Input: A well-based CAF $F$ in $x_{1},\ldots,x_{n}$ with $BP(F)=\{f_{1},\ldots,f_{n}\}$,
where $f_{k}\in\mathbb{\mathbb{Q}}[x_{1},\ldots,x_{k}]\setminus\mathbb{\mathbb{Q}}[x_{1},\ldots,x_{k-1}]$,
$SPT$ and $RTS$ as in the output of Algorithm \ref{SPT}. \\
Output: $ADP$ such that for $1\leq k\leq n$ and for each pair
of adjacent cells $(C,C^{\prime})$ of $D(\Pi_{k}(F))$ with $\dim C^{\prime}=\dim C-1$,
$ADP(C,C^{\prime})$ is a point in $C$ satisfying Condition \ref{ADPCond}.
\begin{enumerate}
\item Let $r$ be the number of real roots of $f_{1}$. For $1\leq i\leq r+1$:

\begin{enumerate}
\item Let $C$ be the $i-th$ $\{f_{1}\}$-sector.
\item If $i>1$ let $C^{\prime}$ be the $i-1$-st $\{f_{1}\}$-section.
We have $SPT(C^{\prime})=PT(v,(),\{1\},\{\})$, $v=(\alpha)=RAV(\Omega_{1})$,
and $\Omega_{1}=((),f_{1},W,\Theta)$. Let $q\in dbl(W)\cap\mathbb{Q}$
and $q>\alpha$. Set $ADP(C,C^{\prime})=(q)$.
\item If $i\leq r$ let $C^{\prime}$ be the $i$-th $\{f_{1}\}$-section.
We have \[
SPT(C^{\prime})=PT(v,(),\{1\},\{\})\]
$v=(\alpha)=RAV(\Omega_{1})$, and \[
\Omega_{1}=((),f_{1},W,\Theta)\]
Let $q\in dbl(W)\cap\mathbb{Q}$ and $q<\alpha$. Set $ADP(C,C^{\prime})=(q)$.
\end{enumerate}
\item For $1\leq k<n$ and for each non-zero-dimensional cell $C$ of $D(\Pi_{k}(F))$:

\begin{enumerate}
\item Let $SPT(C)=PT(v,w,I,J)$ and let $r$ be the number of real roots
of $f_{k+1}$ over $C$. For $1\leq i\leq r+1$:

\begin{enumerate}
\item Let $S$ be the $i-th$ $\{f_{k+1}\}$-sector over $C$.
\item If $i>1$ let $S^{\prime}$ be the $i-1$-st $\{f_{k+1}\}$-section
over $C$. We have $SPT(S^{\prime})=PT(v',w,I\cup\{k+1\},J)$, $v^{\prime}=v\times\alpha=RAV(\Omega_{l+1})$,
and $\Omega_{l+1}=(v,g,W,\Theta)$. Let $q\in dbl(W)\cap\mathbb{Q}$
and $q>\alpha$. Set $ADP(S,S^{\prime})=PT(v,w^{\prime},I,J\cup\{k+1\})$,
where $w^{\prime}=w\times q$.
\item If $i\leq r$ let $S^{\prime}$ be the $i$-th $\{f_{k+1}\}$-section
over $C$. We have $SPT(S^{\prime})=PT(v',w,I\cup\{k+1\},J)$, $v^{\prime}=v\times\alpha=RAV(\Omega_{l+1})$,
and $\Omega_{l+1}=(v,g,W,\Theta)$. Let $q\in dbl(W)\cap\mathbb{Q}$
and $q<\alpha$. Set \[
ADP(S,S^{\prime})=PT(v,w^{\prime},I,J\cup\{k+1\})\]
where $w^{\prime}=w\times q$.
\end{enumerate}
\item For each cell $C^{\prime}$ of $D(\Pi_{k}(F))$ adjacent to $C$ and
such that $\dim C^{\prime}=\dim C-1$:

\begin{enumerate}
\item Let $(a_{1},\ldots,a_{k})=PT(v,w,I,J)=ADP(C,C^{\prime})$ and let
\[RTS(C^{\prime})=((D_{1},\ldots,D_{m}),(k_{1},\ldots,k_{m}),R)\]
\item Let $S_{1}^{\prime},\ldots,S_{s}^{\prime}$ be the$\{f_{k+1}\}$-sections
over $C^{\prime}$, and let $W_{j}^{\prime}$ be the isolating disc
of $\Lambda(SPT(S_{j}^{\prime}))$ for $1\leq j\leq s$.
\item Let $I=\{i_{1},\ldots,i_{l}\}$, $J=\{j_{1},\ldots,j_{k-l}\}$, and
let $g\in\mathbb{Q}[x_{i_{1}},\ldots,x_{i_{l}},x_{k+1}]$ be $f_{k+1}$
with $x_{j_{t}}$ replaced by $a_{j_{t}}$ for $1\leq t\leq k-l$. 
\item Compute $(\Theta;\Omega_{l}^{\prime};v_{1},\ldots,v_{r})=AlgRoots_{l}(v,g,p)$. 
\item For $1\leq i\leq r$ refine the isolating disc $W_{i}$ of $\Lambda(v_{i})$
until it is contained in one of $dbl(W_{1}^{\prime}),\ldots,dbl(W_{s}^{\prime})$
or $W_{i}\cap\Delta(0,R/2)=\emptyset$. Let $S$ be the $i$-th $\{f_{k+1}\}$-section
over $C$. If $W_{i}\subseteq dbl(W_{j}^{\prime})$, set $ADP(S,S_{j}^{\prime})=PT(v_{i},w,I\cup\{k+1\},J)$,
and set $L(i)=S_{j}^{\prime}$. Otherwise if $\Lambda(v_{i})<0$ set
$L(i)=-\infty$ else set $L(i)=\infty$.
\item Set $L(0)=-\infty$ and $L(r+1)=\infty$.
\item For $1\leq i\leq r+1$, let $S$ be the $i$-th $\{f_{k+1}\}$-sector
over $C$. For each $\{f_{k+1}\}$-sector $S^{\prime}$ over $C^{\prime}$
that lies between $L(i-1)$ and $L(i)$ put $u=w\times\Lambda(SPT(S^{\prime}))$
and set $ADP(S,S^{\prime})=PT(v,u,I,J\cup\{k+1\})$.
\end{enumerate}
\end{enumerate}
\item Return $ADP$.
\end{enumerate}
\end{algorithm}
\begin{proof}
Let us now prove correctness of Algorithm \ref{CADAdjacency}. The
working precision $p$ set in step $(1)$ of \emph{SamplePoints} is
used in calls to \emph{AlgRoots}. Since \emph{AlgRoots} raises precision
as needed to reach its goals, $p$ is just an initial value and can
be set arbitrarily e.g. to the number of bits in a double precision
number. Steps $(2)$-$(6)$ construct sample points $SPT(C)$ is all
cells of $D(F)$, starting with sample points in cells of $D(\Pi_{1}(F))$,
and then extending them to sample points in $D(\Pi_{k}(F))$ one coordinate
at a time. An important fact to note is that isolating discs in the
representations of already constructed sample points may change during
the execution of step $(6)$. Namely, in step $(6d)$ the isolating
discs of the coordinates of the sample points $SPT(\Pi_{i}(C))$ for
all projections of the cell $C$ are replaced with their refinements
that were computed in the process of isolating the roots of $f_{k+1}^{C}$.
In particular, for any cell $C\in D(F)$ if $SPT(C)=(a_{1},\ldots,a_{n})$
then for any $1\leq k\leq n$ $SPT(\Pi_{k}(C))=(a_{1},\ldots,a_{k})$
and the isolating discs that appear in the representations of any
algebraic coordinate $a_{i}$ in $SPT(C)$ and in $SPT(\Pi_{k}(C))$
are equal. Note however, that after \emph{SamplePoints} is finished
the representations of $SPT(C)$ are fixed. 

In step $(1)$ of \emph{AdjacencyPoints} for each pair of adjacent
cells $(C,C^{\prime})$ of $D(\Pi_{1}(F))$ with $\dim C^{\prime}=\dim C-1$
the algorithm constructs a point \[
ADP(C,C^{\prime})=(q)\in C\]
such that if $SPT(C^{\prime})=(\alpha)$ and $W$ is the isolating
disc of $\alpha$ then $q\in dbl(W)$. At the start of each iteration
of the loop in step $(2)$ the algorithm has already constructed a
point $ADP(C,C^{\prime})$ for each pair of adjacent cells $(C,C^{\prime})$
of $D(\Pi_{k}(F))$ with $\dim C^{\prime}=\dim C-1$. The points satisfy
Condition \ref{ADPCond}. Steps $(2a)$ and $(2b)$ construct points
$ADP(S,S^{\prime})$ for each pair of adjacent cells $(S,S^{\prime})$
of $D(\Pi_{k+1}(F))$ with $\dim S^{\prime}=\dim S-1$. It is clear
that the constructed points satisfy Condition \ref{ADPCond}. What
we need to show is that the construction will always succeed, pairs
of cells $(S,S^{\prime})$ for which $ADP(S,S^{\prime})$ is constructed
are adjacent, and $ADP(S,S^{\prime})$ is constructed for every pair
of adjacent cells $(S,S^{\prime})$ of $D(\Pi_{k+1}(F))$ with $\dim S^{\prime}=\dim S-1$.
Step $(2a)$ constructs $ADP(S,S^{\prime})$ for every pair of adjacent
cells from a stack over the same cell $C$. Note that in step $(2a)$
we have $\alpha\in W$ and $(dbl(W)\setminus W)\cap\mathbb{R}$ consists
of two intervals, one on each side of $\alpha$, hence we can pick
rational numbers $q\in dbl(W)$ with $q>\alpha$ or $q<\alpha$. If
cells $S$ and $S^{\prime}$ from stacks over different cells $C$
and $C^{\prime}$ are adjacent and $\dim S^{\prime}=\dim S-1$ then
$C$ and $C^{\prime}$ must be adjacent and, by Proposition \ref{WellBasedProp},
$\dim C^{\prime}=\dim C-1$ and $S$ is a section iff $S^{\prime}$
is a section. This shows that step $(2)$ constructs $ADP(S,S^{\prime})$
for every pair of adjacent cells $(S,S^{\prime})$ of $D(\Pi_{k+1}(F))$
with $\dim S^{\prime}=\dim S-1$. 

Let us prove that the construction in step $(2b)$ will always succeed
and pairs of cells $(S,S^{\prime})$ for which $ADP(S,S^{\prime})$
is constructed are adjacent. With notation of step $(2b)$, let \[
(a_{1},\ldots,a_{k})=PT(v,w,I,J)=ADP(C,C^{\prime})\]
and \[
(a_{1}^{\prime},\ldots,a_{k}^{\prime})=PT(v^{\prime},w^{\prime},I^{\prime},J^{\prime})=SPT(C^{\prime})\]
Note that $J^{\prime}\subseteq J$ and $a_{j}^{\prime}=a_{j}$ for
$j\in J^{\prime}$. Let $g^{\prime}=f_{k+1}^{C^{\prime}}$ be $f_{k+1}$
with $x_{j}$ replaced by $a_{j}^{\prime}=a_{j}$ for $j\in J^{\prime}$.
Then $g$ is equal to $g^{\prime}$ with $x_{j}$ replaced by $a_{j}$
for $j\in J\cap I^{\prime}$. For $i\in I^{\prime}$ let $U_{i}^{\prime}$
be the isolating disk of $a_{i}^{\prime}$ in $v^{\prime}$ and let
$U_{i}$ be the isolating disk of $a_{i}^{\prime}$ in the representation
of $\Pi(v^{\prime})$ with which $RTS(C^{\prime})$ was computed.
Note that, by Remark \ref{RefRem}, the current representation of
$v^{\prime}$ is a refinement of the representation with which $RTS(C^{\prime})$
was computed, hence $U_{i}^{\prime}\subseteq dbl(U_{i})$. Hence,
$a_{i}\in dbl(U_{i}^{\prime})$ and $a_{i}\in quad(U_{i})$. By the
condition $(5)$ of Definition \ref{RAV}, for each $1\leq i\leq r$
either $\Lambda(v_{i})$ belongs to one of $W_{1}^{\prime},\ldots,W_{s}^{\prime}$
or $\Lambda(v_{i})\notin\Delta(0,R)$. Therefore we can refine the
isolating disc $W_{i}$ of $\Lambda(v_{i})$ so that it is contained
in one of $dbl(W_{1}^{\prime}),\ldots,dbl(W_{s}^{\prime})$ or $W_{i}\cap\Delta(0,R/2)=\emptyset$.
In the former case the $i$-th $\{f_{k+1}\}$-section over $C$ is
adjacent to the $j$-th $\{f_{k+1}\}$-section over $C^{\prime}$,
in the latter case the $i$-th $\{f_{k+1}\}$-section over $C$ tends
to infinity whose sign is determined by the sign of $\Lambda(v_{i})$.
This shows that sections $(S,S^{\prime})$ for which $ADP(S,S^{\prime})$
is constructed are adjacent. Finally, let $S$ and $S^{\prime}$ be
sectors over $C$ and $C^{\prime}$ defined in step $(2b(vii))$.
Then, by construction in step $(6e)$ of \emph{SamplePoints}, $q=\Lambda(SPT(S^{\prime}))\notin\bigcup_{j=1}^{m}dbl(D_{j})$,
and since $W_{j}^{\prime}\subseteq dbl(D_{j})$ (possibly after reordering
of indices), $q\notin\bigcup_{j=1}^{s}W_{j}^{\prime}$. Moreover,
$-R<q<R$. Since for each $1\leq i\leq r$ either $\Lambda(v_{i})$
belongs to one of $W_{1}^{\prime},\ldots,W_{s}^{\prime}$ or $\Lambda(v_{i})\notin\Delta(0,R)$,
the point $PT(v,u,I,J\cup\{k+1\})$ defined in step $(2b(vii))$ belongs
to $S$. If $ADP(S,S^{\prime})$ is constructed in step $(2b(vii))$
then $S^{\prime}$ is a sector that lies between sections adjacent
to the sections bounding $S$, hence $S$ and $S^{\prime}$ are adjacent. 
\end{proof}

\section{Empirical Results}

An algorithm computing CAD cell adjacencies has been implemented in
C, as a part of the kernel of \emph{Mathematica}. The implementation
takes a quantified system of polynomial equations and inequalities
$S$ and uses \emph{Mathematica} multi-algorithm implementation of
CAD to compute a CAF $F$ such that $D(F)$ is a CAD of $\mathbb{R}^{n}$
consistent with the solution set $A$ of $S$. If $F$ is well-based
the implementation uses Algorithm \ref{ADP} to find the cell adjacencies.
The implementation is geared towards solving a specified topological
problem, e.g. finding the boundary or the connected components of
$A$, hence it avoids computing cell adjacencies for cells that are
known not to belong to the closure of $A$. The current implementation
also works for non-well-based problems in $\mathbb{R}^{3}$ using
ideas from \cite{ACM2} to extend Algorithm \ref{ADP}. 

The experiments have been conducted on a Linux laptop with a $4$-core
$2.7$ GHz Intel Core i7 processor and $16$ GB of RAM. The reported
CPU time is a total from all cores used. For each example we give
three timings. $t_{CAD}$ is the computation time of constructing
a CAF consistent with the solution set the input system. $t_{SP}$
is the time of refining the CAD to a $BP(F)$-invariant CAD of$\mathbb{R}^{n}$
and of constructing sample points in the CAD cells (steps $(1)$-$(6)$
of Algorithm \ref{ADP}). Our implementation refines the CAD while
constructing sample points, which is why we cannot give separate timings.
The third timing, $t_{ADJ}$ is the time of computing cell adjacency
information (steps $(7)$-$(9)$ of Algorithm \ref{ADP}). We also
report the dimension $\dim$ of the embedding space, the number $N_{CELL}$
of cells in the CAD of $A$, the number $N_{ADJ}$ of computed pairs
of adjacent cells whose dimensions differ by one, and the number $N_{CC}$
of connected components of $A$.
\begin{example}
Find cell adjacencies for a CAD of the union of two unit balls in
$\mathbb{R}^{n}$\[
x_{1}^{2}+\ldots+x_{n}^{2}\leq1\vee(x_{1}-1)^{2}+\ldots+(x_{n}-1)^{2}\leq1\]
Note that for $n\leq3$ the balls have full-dimensional intersection,
for $n=4$ they touch at one point, and for $n>4$ they are disjoint.

\begin{table}

\caption{Union of two balls in $\mathbb{R}^{n}$}
\begin{tabular}{|c|c|c|c|c|c|c|}
\hline 
$\dim$ & $t_{CAD}$ & $t_{SP}$ & $t_{ADJ}$ & $N_{CELL}$ & $N_{ADJ}$ & $N_{CC}$\tabularnewline
\hline
\hline 
$2$ & $0.018$ & $0.004$ & $0.001$ & $21$ & $42$ & $1$\tabularnewline
\hline 
$3$ & $0.100$ & $0.033$ & $0.024$ & $179$ & $718$ & $1$\tabularnewline
\hline 
$4$ & $0.489$ & $0.175$ & $0.112$ & $521$ & $3898$ & $1$\tabularnewline
\hline 
$5$ & $1.42$ & $0.773$ & $0.352$ & $954$ & $11910$ & $2$\tabularnewline
\hline 
$6$ & $44.6$ & $24.8$ & $8.92$ & $14050$ & $251758$ & $2$\tabularnewline
\hline
\end{tabular}

\end{table}

\end{example}

\begin{example}
Here we used modified versions of examples from Wilson's benchmark
set \cite{W2} (version 4). Of the $77$ examples we selected $63$
that involved at least $3$ variables and we used quantifier-free
versions of the examples. In $21$ of the examples the system was
not well-based and involved more than $3$ variables, hence our algorithm
did not apply. $7$ examples did not finish in $600$ seconds. Of
the $35$ examples for which our implementation succeeded, $29$ were
well-based and $6$ were not well-based and in $\mathbb{R}^{3}$.
On average, $t_{CAD}$ took $55\%$ of the total time, $t_{SP}$ took
$34\%$, and $t_{ADJ}$ took $11\%$. Five examples with the largest
number of cells are given in Table \ref{wilson}. All but the third
example are well-based.

\begin{table}

\caption{\label{wilson}Wilson's benchmark}
\begin{tabular}{|c|c|c|c|c|c|c|c|}
\hline 
Ex \# & $\dim$ & $t_{CAD}$ & $t_{SP}$ & $t_{ADJ}$ & $N_{CELL}$ & $N_{ADJ}$ & $N_{CC}$\tabularnewline
\hline
\hline 
$2.13$ & $4$ & $0.109$ & $0.263$ & $0.210$ & $3104$ & $10576$ & $1$\tabularnewline
\hline 
$2.16$ & $3$ & $3.06$ & $2.65$ & $1.27$ & $2811$ & $37416$ & $1$\tabularnewline
\hline 
$6.1$ & $3$ & $0.768$ & $0.794$ & $0.312$ & $2774$ & $8926$ & $2$\tabularnewline
\hline 
$5.10$ & $4$ & $14.9$ & $11.2$ & $4.01$ & $2256$ & $63190$ & $1$\tabularnewline
\hline 
$6.6$ & $6$ & $14.6$ & $7.85$ & $3.52$ & $2128$ & $76360$ & $1$\tabularnewline
\hline
\end{tabular}%
\end{table}

\end{example}

\begin{example}
Here we took the $32$ 3D solids that appear in \emph{Mathematica}
\noun{SolidData} and intersected each of them with the solution set
of $9(x+y+z)^{2}>z^{2}+1$. All $32$ examples were well-based and
in all our implementation succeeded. On average, $t_{CAD}$ took $37\%$
of the total time, $t_{SP}$ took $47\%$, and $t_{ADJ}$ took $16\%$.
The five solids which resulted in the largest number of cells are:
\begin{enumerate}
\item Steinmetz 6-solid\begin{eqnarray*}
 & 2x^{2}+(y-z)^{2}\leq2\wedge2x^{2}+(y+z)^{2}\leq2\wedge\\
 & 2y^{2}+(x-z)^{2}\leq2\wedge2y^{2}+(x+z)^{2}\leq2\wedge\\
 & (x-y)^{2}+2z^{2}\leq2\wedge(x+y)^{2}+2z^{2}\leq2\end{eqnarray*}

\item Sphericon\begin{eqnarray*}
 & (x^{2}+y^{2}\leq(|z|-1)^{2}\wedge x\geq0\wedge-1\leq z\leq1)\vee\\
 & (x^{2}+z^{2}\leq(|y|-1)^{2}\wedge x\leq0\wedge-1\leq y\leq1)\end{eqnarray*}

\item Steinmetz 4-solid\begin{eqnarray*}
 & x^{2}+y^{2}\leq1\wedge9x^{2}+y^{2}+8z^{2}\leq9+164/29yz\wedge\\
 & 3x^{2}+284/41xy+7y^{2}+82/29yz+8z^{2}\leq9+49/10xz\wedge\\
 & 3x^{2}+7y^{2}+49/10xz+82/29yz+8z^{2}\leq9+284/41xy\end{eqnarray*}

\item Solid capsule\begin{eqnarray*}
 & x^{2}+y^{2}+(z-1/2)^{2}\leq1\vee x^{2}+y^{2}+(z+1/2)^{2}\leq1\vee\\
 & -1/2\leq z\leq1/2\wedge x^{2}+y^{2}\leq1\end{eqnarray*}

\item Reuleaux tetrahedron\begin{eqnarray*}
 & x^{2}+y^{2}+(19/31+z)^{2}\leq1\wedge\\
 & (x-15/26)^{2}+y^{2}+(z-9/44)^{2}\leq1\wedge\\
 & (x+11/38)^{2}+(y-1/2)^{2}+(z-9/44)^{2}\leq1\wedge\\
 & (x+11/38)^{2}+(y+1/2)^{2}+(z-9/44)^{2}\leq1\end{eqnarray*}

\end{enumerate}
The details are given in Table \ref{solids}. 

\begin{table}
\caption{\label{solids}Intersections of solids with $9(x+y+z)^{2}>z^{2}+1$. }
\begin{tabular}{|c|c|c|c|c|c|c|c|}
\hline 
Ex \# & $\dim$ & $t_{CAD}$ & $t_{SP}$ & $t_{ADJ}$ & $N_{CELL}$ & $N_{ADJ}$ & $N_{CC}$\tabularnewline
\hline
\hline 
$1$ & $3$ & $254$ & $614$ & $189$ & $156688$ & $4320078$ & $2$\tabularnewline
\hline 
$2$ & $3$ & $59.5$ & $82.6$ & $27.0$ & $54256$ & $767462$ & $2$\tabularnewline
\hline 
$3$ & $3$ & $52.1$ & $78.2$ & $20.1$ & $24476$ & $461614$ & $2$\tabularnewline
\hline 
$4$ & $3$ & $8.11$ & $5.42$ & $3.06$ & $17152$ & $84162$ & $2$\tabularnewline
\hline 
$5$ & $3$ & $47.9$ & $53.5$ & $15.1$ & $11756$ & $349976$ & $2$\tabularnewline
\hline
\end{tabular}%
\end{table}
\end{example}

\bibliographystyle{abbrv}

\end{document}